\documentclass[10pt]{amsart}
\usepackage{amsmath,amsfonts,amssymb,mathrsfs}

\title[Relativistic Boltzmann equation]
{\rm Asymptotic behaviour of the relativistic Boltzmann equation in the
Robertson-Walker spacetime}

\author{Ho Lee}
\address{Department of Mathematics and Research Institute for Basic Science,
Kyung Hee University, Seoul, 130-701, Republic of Korea}
\email{holee@khu.ac.kr}

\begin{document}

\newtheorem{theorem}{Theorem}[section]
\newtheorem{lemma}{Lemma}[section]
\newtheorem{corollary}{Corollary}[section]
\newtheorem{proposition}{Proposition}[section]
\newtheorem{remark}{Remark}[section]
\newtheorem{definition}{Definition}[section]

\renewcommand{\theequation}{\thesection.\arabic{equation}}
\renewcommand{\thetheorem}{\thesection.\arabic{theorem}}
\renewcommand{\thelemma}{\thesection.\arabic{lemma}}
\newcommand{\bbr}{\mathbb R}
\newcommand{\bbs}{\mathbb S}
%\newcommand{\+}{|\!|\!|}
%\def\charf {\mbox{{\text 1}\kern-.24em {\text l}}}
%\allowdisplaybreaks[3]

\begin{abstract}
In this paper, we study the relativistic Boltzmann equation in the spatially flat
Robertson-Walker spacetime. For a certain class of scattering kernels,
global existence of classical solutions is proved. We use the standard method of
Illner and Shinbrot for the global existence and apply the splitting technique
of Guo and Strain for the regularity of solutions. The main interest of this paper
is to study the evolution of matter distribution, rather than the evolution of spacetime.
We obtain the asymptotic behaviour of solutions and will understand
how the expansion of the universe affects the evolution of matter distribution.
\end{abstract}

\maketitle

\section{Introduction}\setcounter{equation}{0}
One of the simplest ways to describe an expanding universe is to consider the spatially
homogeneous and isotropic spacetimes with suitable matter models.
This kind of spacetimes are called the Robertson-Walker(RW) spacetimes, and in this paper
we are interested in a kinetic matter model. One may consider the Vlasov equation as
a kinetic matter model to get the Einstein-Vlasov system, and this system of equations has been
extensively studied for last several decades. In this case, more general spacetime models
such as the spacetimes of Bianchi types can be considered, and many interesting
results can be found in the literature, for instance see
\cite{A11,L04,N10,R,R94,R95,R96,R05,RT99,RU00}.
In this paper, we only consider the spatially flat RW spacetime with the metric
\[
ds^2=-dt^2+R^2(t)(dx^2+dy^2+dz^2),
\]
and as a matter model we will consider the Boltzmann equation.
Here, the scale factor $R$ will be assumed to be given.
The main purpose of this paper is to study the evolution of matter distribution,
rather than the evolution of spacetime, and in this sense we study the asymptotic
behaviour of solutions of the Boltzmann equation in the RW spacetime.

In the kinetic theory, matter is treated as a collection of particles, and the
Boltzmann equation takes into account the effect of collisions between particles.
It is well known that collisions between particles
lead to an interesting phenomenon, which is not observed
in the Vlasov case, such that almost every solution of the Boltzmann equation
converges to an equilibrium state. This can be proved in a
mathematically rigorous way, and indeed in the Newtonian and
the Minkowski cases innumerable references can be found on this subject.
We only refer to \cite{CIP,G} and their references for the Newtonian Boltzmann equation
and \cite{GS12,S101,S102,S11} for recent results on the relativistic Boltzmann
equation in the Minkowski spacetime.
In this paper, we will consider the relativistic Boltzmann equation
in the RW spacetime. The main interest of this paper is to investigate
how the expansion of the universe affects the asymptotic behaviour of solutions of
the Boltzmann equation. Recently in \cite{L13}, this question was answered in the context of
Newtonian cosmology. The main result of \cite{L13} can be stated as follows: depending on
the rate of growth of the scale factor, solutions of the Boltzmann equation may or may not
approach equilibrium states. In this paper, we will answer this question in the framework
of general relativity.

The Boltzmann equation in general relativity has not been studied much.
Local existence was proved by Bancel and Choquet-Bruhat many years ago \cite{B73,BCB73},
and Noutchegueme and his colleagues have studied the Boltzmann equation in some
cosmological settings with strong assumptions on
scattering kernels \cite{ND06,NDT05,NT05,T09}.
A hard potential case has been recently studied in \cite{LR131}
in the RW spacetime. In this paper, we will also study the Boltzmann equation
in the RW spacetime, but we will consider a different type of scattering kernels
with a different type of existence proof.
In the following two subsections, we will review basic information
on the Boltzmann equation in the RW spacetime. In Section \ref{Sec preliminaries}, we introduce
a change of variables to write the Boltzmann equation in a simple form, and then we make
the main assumptions on the transformed equation.
After establishing several pointwise estimates, we prove global existence
of solutions in Section \ref{Sec global existence}.

\subsection{The Boltzmann equation}
The relativistic Boltzmann equation in the RW spacetime can be written as follows:
\begin{gather}\label{boltzmann in RW}
\partial_tf-2\frac{\dot{R}}{R}\sum_{i=1}^3p^i\partial_{p^i}f=Q(f,f),
\end{gather}
where $f=f(t,p)$ is the distribution function, $\dot{R}$ denotes the time derivative of $R$,
and $Q$ is called the Boltzmann collision operator. It can be written as
\begin{gather}\label{collision op in RW}
Q(f,f)=R^3\int_{\bbr^3}\int_{\bbs^2}v_\phi\sigma(g,\omega)\Big(
f(p')f(q')-f(p)f(q)\Big)\,d\omega\,dq,\quad v_\phi=\frac{g\sqrt{s}}{p^0q^0},
\end{gather}
where $v_\phi$ and $\sigma$ are called the M{\o}ller velocity and the scattering kernel
respectively, and the scalar quantities $g$ and $s$ are defined by
\begin{gather}\label{g and s}
g=\sqrt{(p_\alpha-q_\alpha)(p^\alpha-q^\alpha)}\quad\text{and}\quad
s=-(p_\alpha+q_\alpha)(p^\alpha+q^\alpha).
\end{gather}
Here, $p^\alpha$ and $q^\alpha$ denote four-dimensional momentum variables, and
$p$ and $q$ in the equation \eqref{boltzmann in RW} denote the spatial components
of $p^\alpha$ and $q^\alpha$, i.e.
\[
p^\alpha=(p^0,p^1,p^2,p^3)\quad\text{and}\quad
p=(p^1,p^2,p^3).
\]
Latin indices will be assumed to run from $1$ to $3$ as in \eqref{boltzmann in RW},
while Greek indices run from $0$ to $3$ as in \eqref{g and s}, and the Einstein
summation convention will be assumed. The indices are lowered through the metric considered,
i.e. $p_\alpha=g_{\alpha\beta}p^\beta$, and in the RW case we have $p_0=-p^0$ and $p_k=R^2p^k$.
Assuming that all the particles have the same mass, we get the mass shell condition
\begin{equation}\label{mass shell}
p_\alpha p^\alpha=-m^2=-1,
\end{equation}
where we assumed $m=1$ for simplicity, and in the RW case it
can be written as
\[
p^0=\sqrt{1+R^2|p|^2},
\]
where $|\cdot|$ denotes the modulus of $p$ in $\bbr^3$. The momentum variables with primes
like $p'$ and $q'$ in \eqref{collision op in RW} denote post-collisional momenta
for given pre-collisional momenta $p$ and $q$. In the following section, we briefly review
several different types of representations of post-collisional momenta.
For more details on the relativistic kinetic equations, we refer to
\cite{A11,dvv,E,LR13,R,St}.

\subsection{Post-collisional momenta}
Suppose that two particles having momenta $p^\alpha$ and $q^\alpha$ collide,
and let $p'^\alpha$ and $q'^\alpha$ be their momenta after the collision.
In the collision process it is assumed that their total energy and momentum are conserved,
which can be written by
\[
p'^\alpha+q'^\alpha=p^\alpha+q^\alpha.
\]
Due to this energy-momentum conservation, the post-collisional momenta
can be parametrized by $p^\alpha$ and $q^\alpha$ with some additional
parameters. In the nonrelativistic case, the following two different types of representations
of post-collisional momenta are easily found. For given three dimensional vectors $\xi$
and $\xi_*$, we have
\begin{gather}
\xi'=\frac{\xi+\xi_*}{2}+\frac{|\xi-\xi_*|}{2}\sigma,\quad
\xi_*'=\frac{\xi+\xi_*}{2}-\frac{|\xi-\xi_*|}{2}\sigma,\quad\sigma\in\bbs^2,\label{sigma nr}\\
\text{or}\qquad
\xi'=\xi-((\xi-\xi_*)\cdot\omega)\omega,\quad
\xi_*'=\xi_*+((\xi-\xi_*)\cdot\omega)\omega,\quad\omega\in\bbs^2,\label{omega nr}
\end{gather}
where $\cdot$ and $|\cdot|$ are the usual inner product and the corresponding norm
in $\bbr^3$,
and they are sometimes called the $\sigma$-representation and the $\omega$-representation
respectively (pages 126--127 of \cite{V}).
To find a relativistic analogue of \eqref{sigma nr}, we make
an assumption that the
post-collisional momenta have the form of
\begin{gather}
p'^\alpha=\frac{p^\alpha+q^\alpha}{2}+\frac{g}{2}\Omega^\alpha,\quad
q'^\alpha=\frac{p^\alpha+q^\alpha}{2}-\frac{g}{2}\Omega^\alpha,\label{sigma r}
\end{gather}
for some four-vector $\Omega^\alpha$. In a similar way, we may write
\begin{gather}\label{omega r}
p'^\alpha=p^\alpha-\Big((p_\beta-q_\beta)\Omega^\beta\Big)\Omega^\alpha,\quad
q'^\alpha=q^\alpha+\Big((p_\beta-q_\beta)\Omega^\beta\Big)\Omega^\alpha,
\end{gather}
for some different four-vector $\Omega^\alpha$.
Then, the mass shell condition \eqref{mass shell} gives
two constraints on the parameter $\Omega^\alpha$.
In the case of \eqref{sigma r}, we have
\[
-1=-\frac{s}{4}\pm\frac{g^2}{4}n_\alpha \Omega^\alpha+\frac{g^2}{4}\Omega_\alpha\Omega^\alpha,
\]
where $n^\alpha=p^\alpha+q^\alpha$, and equivalently we get
two constraint equations:
\begin{equation}\label{constraint}
n_\alpha\Omega^\alpha=0\quad\mbox{and}\quad\Omega_\alpha\Omega^\alpha=1.
\end{equation}
In the case of \eqref{omega r},
the same constraint equations are obtained.
Hence, the parameter $\Omega^\alpha$ reduces to a four-vector which can be represented
by two independent variables, and we take
$\omega\in\bbs^2$ as usual. The simplest way to find $\Omega^\alpha$ is
to construct a vector $t^\alpha$ which is orthogonal to $n^\alpha$ and then normalize it.
Using the parameter $\omega\in\bbs^2$, we find
$t^\alpha=(n_i\omega^i,-n_0\omega)$ and normalize it to get
\begin{equation}\label{Omega}
\Omega^\alpha=\frac{t^\alpha}{\sqrt{t_\beta t^\beta}}\quad\mbox{with}\quad
t^\alpha=(n_i\omega^i,-n_0\omega).
\end{equation}
It is easy to see that $n_\alpha t^\alpha=0$ and $\Omega_\alpha \Omega^\alpha=1$.
Plugging \eqref{Omega} into
\eqref{sigma r} or \eqref{omega r},
we obtain two parametrizations
of $p'^\alpha$ and $q'^\alpha$ in terms of $p^\alpha$, $q^\alpha$, and $\omega$.
We note that the parameter $\Omega^\alpha$ is basically the same with the
one that was previously found by
Glassey and Strauss \cite{GS91,GS93} in the Minkowski case.
It is written in the Minkowski space as
\begin{equation}\label{Omega M}
\Omega^\alpha_{M}=\frac{t^\alpha}{\sqrt{t_\beta t^\beta}}\quad\mbox{with}\quad
t^\alpha=(n\cdot\omega,n^0\omega),
\end{equation}
where $n\cdot\omega$ is the usual inner product of $n$ and $\omega$ in $\bbr^3$.
We combine \eqref{omega r} and \eqref{Omega M} to get
the parametrization of \cite{GS91,GS93}.

In the Minkowski case we can find a different form of the parameter $\Omega^\alpha$.
To find a vector orthogonal to $n^\alpha$, we have constructed $t^\alpha$ in \eqref{Omega}
and then normalized it.
However, we may decompose $\omega=\omega_1+\omega_2$
such that $n\cdot\omega_1=0$ and $n\cdot\omega_2=n\cdot\omega$, for instance
\begin{equation}\label{decomposition of omega}
\omega_1=\omega-\frac{(n\cdot\omega)n}{|n|^2}\quad\mbox{and}\quad
\omega_2=\frac{(n\cdot\omega)n}{|n|^2}.
\end{equation}
Then, $t^\alpha=t^\alpha(\omega)$, as a linear function of $\omega$, can also be decomposed as
\[
t^\alpha=t^\alpha(\omega)=t^\alpha(\omega_1)+t^\alpha(\omega_2)=:t_1^\alpha+t_2^\alpha.
\]
We can observe
\[
n_{\alpha} t_1^\alpha=
n_{\alpha} t_2^\alpha=
t_{1\alpha}^{\phantom\alpha} t_2^\alpha=0,
\]
and
\[
t_{1\alpha}^{\phantom\alpha} t_1^\alpha=(n^0)^2|\omega_1|^2,\quad
t_{2\alpha}^{\phantom\alpha} t_2^\alpha=s|\omega_2|^2.
\]
Since $|\omega_1|^2+|\omega_2|^2=1$, we find a different form of $\Omega^\alpha$ as
\begin{equation}\label{Omega S}
\Omega_{S}^\alpha=\frac{1}{n^0}t_1^\alpha+\frac{1}{\sqrt{s}}t_2^\alpha\quad\mbox{with}\quad
t_k^\alpha=(n\cdot\omega_k,n^0\omega_k),\quad k=1,2.
\end{equation}
It is easy to see that \eqref{Omega S} satisfies the constraints \eqref{constraint}. We remark
that the representation combining \eqref{sigma r} and \eqref{Omega S} is basically the one that
Strain has derived and used in a series of his papers, for instance \cite{GS12,S101,S102,S11}.
It has turned out that using the parameter \eqref{Omega S} is crucial in the proof of existence of
regular solutions to the relativistic Boltzmann equation.
We refer to \cite{GS12,S11} for more detailed explanation of the
parametrization of this type.

In this paper, we consider the RW spacetime.
In this case the parameter $\Omega^\alpha$ in \eqref{Omega} reduces to
the following:
\begin{equation}\label{Omega R}
\Omega^\alpha_{R}=\frac{t^\alpha}{\sqrt{t_\beta t^\beta}}\quad\mbox{with}\quad
t^\alpha=(R^2n\cdot\omega,n^0\omega).
\end{equation}
On the other hand,
the parameter $\Omega^\alpha_S$ in \eqref{Omega S}
is only valid in the Minkowski case, hence we
generalize it to the RW case as follows:
\begin{equation}\label{Omega RS}
\Omega^\alpha_{RS}=\frac{1}{Rn^0}t^\alpha_1+\frac{1}{R\sqrt{s}}t^\alpha_2\quad\mbox{for}\quad
t_k^\alpha=(R^2n\cdot\omega_k,n^0\omega_k),\quad k=1,2,
\end{equation}
where $\omega_1$ and $\omega_2$ are given by \eqref{decomposition of omega}.
In the present paper we only use the representation \eqref{sigma r},
but for the parameter $\Omega^\alpha$ we use both of \eqref{Omega R} and \eqref{Omega RS}
as in \cite{GS12}. Since we are interested in classical solutions, we need to control
derivatives of post-collisional momenta,
and the splitting technique and interplay between
\eqref{Omega R} and \eqref{Omega RS} proposed in \cite{GS12}
will efficiently control them.

\section{Preliminaries}\label{Sec preliminaries}\setcounter{equation}{0}
In this section, we first
introduce a change of variables so that the Boltzmann equation in the RW spacetime
is written in a simple form, and then
we make the main assumptions on the transformed equation.
At the end of the section,
we collect elementary lemmas which can be easily proved by direct calculations.

\subsection{Change of variables}
In this paper we will consider the RW spacetimes,
and in this case the Boltzmann equation is written in a simple form
if we use covariant variables.
To be explicit, the distribution
function $f$ will be considered as a function of $t$ and
\[
p_k=g_{k\beta}p^\beta=R^2p^k,\quad k=1,2,3.
\]
This argument was previously used in \cite{LR131}, but in this paper
for simplicity, instead of using lower indices, we introduce a new variable $v$ such that
\begin{equation}\label{v}
v=(v^1,v^2,v^3),\quad v^k=R^2p^k,\quad v^0=\sqrt{1+R^{-2}|v|^2}=p^0,
\end{equation}
where $|v|$ is the modulus of $v$ in $\bbr^3$. The post-collisional momentum $p'^\alpha$
in the representation of \eqref{sigma r} and \eqref{Omega R} is written as follows:
\begin{align*}
p'^0&=\frac{p^0+q^0}{2}
+\frac{g}{2}\frac{R^2n\cdot\omega}{\sqrt{R^2(n^0)^2-R^4(n\cdot\omega)^2}}
=\frac{v^0+u^0}{2}+\frac{g}{2R}\frac{\tilde{n}\cdot\omega}{\sqrt{(\tilde{n}^0)^2
-R^{-2}(\tilde{n}\cdot\omega)^2}},\\
p'^k&=\frac{p^k+q^k}{2}
+\frac{g}{2}\frac{n^0\omega^k}{\sqrt{R^2(n^0)^2-R^4(n\cdot\omega)^2}}
=\frac{v^k+u^k}{2R^2}+\frac{g}{2R}\frac{\tilde{n}^0\omega^k}{\sqrt{(\tilde{n}^0)^2
-R^{-2}(\tilde{n}\cdot\omega)^2}},
\end{align*}
where $\tilde{n}=v+u$ and $\tilde{n}^0=v^0+u^0$,
and we may write
\begin{equation}\label{change of variables}
v'^k=R^2p'^k,\quad u'^k=R^2q'^k,\quad v'^0=p'^0,\quad u'^0=q'^0.
\end{equation}
In a similar way the post-collisional momentum in \eqref{sigma r} and \eqref{Omega RS}
is written as
\begin{align*}
p'^0&=\frac{p^0+q^0}{2}
+\frac{Rg}{2}\frac{1}{\sqrt{s}}(n\cdot\omega_2)
=\frac{v^0+u^0}{2}+\frac{g}{2R}\frac{1}{\sqrt{s}}(\tilde{n}\cdot\omega),\\
p'^k&=\frac{p^k+q^k}{2}
+\frac{g}{2R}\left(\omega_1+\frac{n^0\omega_2}{\sqrt{s}}\right)
=\frac{v^k+u^k}{2R^2}+\frac{g}{2R}
\left(\left(\omega-\frac{(\tilde{n}\cdot\omega)\tilde{n}}{|\tilde{n}|^2}\right)
+\frac{\tilde{n}^0}{\sqrt{s}}\frac{(\tilde{n}\cdot\omega)\tilde{n}}{|\tilde{n}|^2}\right),
\end{align*}
and we write $v'$ and $u'$ as in \eqref{change of variables}.
Throughout the paper we will only consider the new variables $v$ and $u$, hence we may drop the
tildes to define
\begin{gather}
n^0:=v^0+u^0\quad\mbox{and}\quad n:=v+u,\label{n^alpha}
\end{gather}
and in the representation of \eqref{sigma r} and \eqref{Omega R} we have
\begin{gather}
\left(
\begin{aligned}
&v'^0\\
&v'^k
\end{aligned}
\right)
=
\left(
\begin{aligned}
&\frac{v^0+u^0}{2}+\frac{g}{2R}\frac{n\cdot\omega}{\sqrt{(n^0)^2
-R^{-2}(n\cdot\omega)^2}}\\
&\frac{v^k+u^k}{2}+\frac{Rg}{2}\frac{n^0\omega^k}{\sqrt{(n^0)^2
-R^{-2}(n\cdot\omega)^2}}
\end{aligned}
\right),\quad\omega\in\bbs^2,\label{v' R}
\end{gather}
while in the representation of \eqref{sigma r} and \eqref{Omega RS} we have
\begin{gather}
\left(
\begin{aligned}
&v'^0\\
&v'^k
\end{aligned}
\right)
=
\left(
\begin{aligned}
&\frac{v^0+u^0}{2}+\frac{g}{2R}
\frac{1}{\sqrt{s}}(n\cdot\hat{\omega})\\
&\frac{v^k+u^k}{2}+\frac{Rg}{2}
\left(\left(\hat{\omega}^k-\frac{(n\cdot\hat{\omega})n^k}{|n|^2}\right)
+\frac{n^0}{\sqrt{s}}\frac{(n\cdot\hat{\omega})n^k}{|n|^2}\right)
\end{aligned}
\right),\quad\hat{\omega}\in\bbs^2,\label{v' RS}
\end{gather}
where $\omega$ and $\hat{\omega}$ denote unit vectors in $\bbr^3$.

We note that $v'$ is written in \eqref{v' R} and \eqref{v' RS} with different unit vectors
$\omega$ and $\hat{\omega}$ on $\bbs^2$.
For a fixed $v'$, we may compare \eqref{v' R} and \eqref{v' RS} to find a relation between
$\omega$ and $\hat{\omega}$. Let $A$ and $B^k$ be the quantities in \eqref{v' R} such that
\[
v'^0=\frac{v^0+u^0}{2}+\frac{g}{2R}A,\quad
v'^k=\frac{v^k+u^k}{2}+\frac{Rg}{2}B^k,
\]
and then $\hat{\omega}$ is determined by comparing \eqref{v' R} and \eqref{v' RS} such that
\begin{align*}
\frac{1}{\sqrt{s}}(n\cdot\hat{\omega})=A,\quad
\left(\hat{\omega}^k-\frac{(n\cdot\hat{\omega})n^k}{|n|^2}\right)
+\frac{n^0}{\sqrt{s}}\frac{(n\cdot\hat{\omega})n^k}{|n|^2}=B^k.
\end{align*}
To be explicit, we plug the first one to the second one to get
\begin{align}\label{relation between omega}
\hat{\omega}^k&=B^k+\frac{\sqrt{s}An^k}{|n|^2}-\frac{n^0An^k}{|n|^2}\\
&=\frac{1}{\sqrt{(n^0)^2-R^{-2}(n\cdot\omega)^2}}
\left(n^0\omega^k+\frac{\sqrt{s}(n\cdot\omega)n^k}{|n|^2}-\frac{n^0(n\cdot\omega)n^k
}{|n|^2}\right).\nonumber
\end{align}
We observe that $\omega=\hat{\omega}$ in two special cases.
Multiplying ${\omega}$ to \eqref{relation between omega}, we obtain
\begin{align*}
\omega\cdot\hat{\omega}
&=\frac{1}{\sqrt{(n^0)^2-R^{-2}(n\cdot\omega)^2}}
\left(n^0+\frac{\sqrt{s}(n\cdot{\omega})^2}{|n|^2}
-\frac{n^0(n\cdot\omega)^2}{|n|^2}\right),
\end{align*}
and this quantity is unity when
$n\cdot{\omega}=0$ or $(n\cdot{\omega})^2=|n|^2$.
In the other cases the unit vectors $\omega$ and $\hat{\omega}$ will in general have different
values. For instance, we may consider from \eqref{relation between omega}
\[
n\cdot\hat{\omega}=\frac{\sqrt{s}}{\sqrt{(n^0)^2-R^{-2}(n\cdot\omega)^2}}
(n\cdot\omega),
\]
which shows that $n\cdot\omega$ and $n\cdot\hat{\omega}$ have the same sign and satisfy
\begin{align}\label{relation between omega 2}
|n\cdot{\omega}|\geq |n\cdot\hat{\omega}|,
\end{align}
where we used \eqref{elementary 6} of Lemma \ref{Lem elementary}.
With these variables the Boltzmann equation is written as
\begin{equation}\label{boltzmann}
\partial_tf=R^{-3}\iint v_\phi\sigma(g,{\omega})
\Big(f(v')f(u')-f(v)f(u)\Big)\,d{\omega}\,du,\quad v_\phi=\frac{g\sqrt{s}}{v^0u^0},
\end{equation}
where $v'$ and $u'$ are parametrized by \eqref{v' R} or \eqref{v' RS}.
The scalar quantities $g$ and $s$ are understood as functions of $v$ and $u$, i.e.
\begin{equation}\label{g and s in v}
g=\sqrt{-(v^0-u^0)^2+R^{-2}|v-u|^2}
=\sqrt{-2-2v^0u^0+2R^{-2}(v\cdot u)},\quad
s=4+g^2.
\end{equation}
Note that the equation \eqref{boltzmann}
is equivalent to \eqref{boltzmann in RW} for classical solutions. In this paper
the Boltzmann equation will refer to the equation \eqref{boltzmann}.

\subsection{Main assumptions}
To prove existence of solutions of the equation \eqref{boltzmann},
we will follow the standard argument of Illner and Shinbrot \cite{IS84},
where existence of small solutions of the Boltzmann equation
was proved in the Newtonian case.
The original idea has been applied to several different physical settings,
for instance see \cite{G06,G01,S101}, where the Vlasov-Poisson-Boltzmann system and
the relativistic Boltzmann equation in the Minkowski spacetime have been studied.
The followings are our main assumptions in this paper.
We first assume that the scale factor is given as an increasing function, and then
choose a suitable weight function. We also make assumptions on the scattering kernel,
and the relevance of the assumption will be discussed.
\bigskip

\noindent{\bf Scale factor.} We assume that the scale factor $R$ satisfies
\begin{equation}\label{scale factor}
R(0)=1\quad\text{and}\quad R'(t)\geq 0\quad\text{with}\quad \lim_{t\to\infty}R(t)=\infty.
\end{equation}

\noindent{\bf Weight function.}
We choose the weight function as $e^{|v|^2}$ and define
\begin{equation}\label{norm}
\|f(t)\|:=\sup\left\{
\left|e^{|v|^2}\partial_{v^k}^jf(\tau,v)\right|:
0\leq \tau\leq t,~v\in\bbr^3,~j=0,1,~k=1,2,3\right\}.
\end{equation}

\noindent{\bf Scattering kernel.}
We assume that the scattering kernel satisfies the following conditions:
for some positive constant $A$ and $0\leq b<3$,
\begin{align}\label{scattering 1}
0\leq\sigma(g,\omega)\leq A(1+g^{-b})\sigma_0(\omega)\quad\text{and}\quad
|\partial_g\sigma(g,\omega)|\leq Ag^{-b-1}\sigma_0(\omega),
\end{align}
where $\sigma_0$ is bounded and supported in $\bbs^2_R$ such that
\begin{equation}\label{scattering 2}
0\leq\sigma_0(\omega)\leq\sigma_1\mbox{\bf 1}_{\bbs^2_R}(\omega)
\end{equation}
for some positive constanat $\sigma_1$. Here,
$\mbox{\bf 1}_{\bbs^2_R}(\omega)$ is the indicator function of the set
$\bbs^2_R$, which is defined as follows:
for some positive constant $B$, we define
\[
\bbs^2_R:=\left\{\omega\in \bbs^2:
\frac{|v-u|^2|n\times\omega|^2}{2R^2s}\leq B\right\},
\]
where $n$ and $s$ are defined by \eqref{n^alpha} and \eqref{g and s in v} respectively.

We introduced a cutoff set $\bbs^2_R$ on the angular part of the scattering kernel.
It depends on $v$, $u$, and $t$, but we can see that for each $v$ and $u$ the cutoff set
$\bbs^2_R$ converges to $\bbs^2$, since $s$ has a lower bound as in \eqref{elementary 1}
and $R$ is an increasing function of $t$. More explicitly, we can find a finite $t_0$
such that we have $\bbs^2_R=\bbs^2$ for $t\geq t_0$, hence the restriction
disappears for large $t$. A similar restriction on the scattering kernel can be found
in a different physical situation as in \cite{S101}, where the author studied the Newtonian
limit of the Boltzmann equation and the speed of light $c$ was treated as a parameter.
To first prove global existence
of solutions, the author introduced a cutoff set $\mathcal{B}_c$
and showed that the cutoff set $\mathcal{B}_c$ converges to $\bbs^2$ for large $c$
(see Lemma 3.1 of \cite{S101}).
This restriction to $\mathcal{B}_c$ was crucial in the proof of the existence theorem
of \cite{S101}, and in the present paper
the cutoff set $\bbs^2_R$ will play the role similar to $\mathcal{B}_c$.
In Section \ref{Sec weight function}, we will see that the weight function $e^{|v|^2}$
works well in our case under the restriction of \eqref{scattering 2}.

In this paper, we will use two different representations of post-collisional momenta,
and then for a given $v'$ two different unit vectors $\omega$ and $\hat{\omega}$
will be used in the representations \eqref{v' R} and \eqref{v' RS}, respectively.
Therefore, the scattering kernel $\sigma$ in principle should be given in
different forms, for instance $\sigma(g,\omega)$ and $\hat{\sigma}(g,\hat{\omega})$ respectively.
However, in this paper we will assume that both $\sigma$ and $\hat{\sigma}$ satisfy
\eqref{scattering 1} and \eqref{scattering 2} and will not distinguish $\sigma$ and
$\hat{\sigma}$. Note that $\omega$ and $\hat{\omega}$ are related to each other by
\eqref{relation between omega}, and we can observe from \eqref{relation between omega 2} that
more post-collisional momenta are cut off in the case of \eqref{v' RS}
under the restriction \eqref{scattering 2}.

%%%%%%%%%%%%%%%%%%%%%%%%%%%%%%%%%%%%%%%%%%%%%%%%%%%%%%%%%%%%%%%%%%%%%%%%%%%%%%%%%%%%%%%%%%%%%%%%%
\subsection{Basic lemmas}
In this part we collect elementary lemmas which are proved by direct calculations.
The proofs of the following lemmas are almost same with the Minkowski case,
but we present them for the reader's convenience.

\begin{lemma}\label{Lem elementary}
The following estimates hold for the quantities defined in the previous sections:
\begin{gather}
s=4+g^2,\quad
2\leq \sqrt{s},\quad
g\leq\sqrt{s},\label{elementary 1}\\
g\leq\sqrt{s}\leq 2\sqrt{v^0u^0},\label{elementary 2}\\
\frac{|v-u|}{\sqrt{v^0u^0}}\leq Rg\leq |v-u|,\label{elementary 3}\\
|v|\leq Rv^0,\quad v^0=\sqrt{1+R^{-2}|v|^2}\leq
\sqrt{1+|v|^2},\label{elementary 4}\\
Rg=|v-u|\sqrt{1-\frac{|v+u|^2\cos^2\theta_0}{R^2(v^0+u^0)^2}},\label{elementary 5}\\
\sqrt{s}\leq\sqrt{(n^0)^2-R^{-2}(n\cdot{\omega})^2},\label{elementary 6}\\
\sqrt{s}\geq \max\left\{\sqrt{\frac{v^0}{u^0}},\sqrt{\frac{u^0}{v^0}}\right\},
\label{elementary 7}
\end{gather}
where $\theta_0$ is the angle between $v+u$ and $v-u$.
\end{lemma}
\begin{proof}
Since $s=-(p_\alpha+q_\alpha)(p^\alpha+q^\alpha)$
and $g=\sqrt{(p_\alpha-q_\alpha)(p^\alpha-q^\alpha)}$, we have
\[
s=2-2p_\alpha q^\alpha=4-2-2p_\alpha q^\alpha=4+g^2.
\]
The inequalities $2\leq\sqrt{s}$ and $g\leq\sqrt{s}$ are now clear.
The second inequality of \eqref{elementary 2} is obtained by
\begin{align*}
s&=2-2p_\alpha q^\alpha=2+2p^0q^0-2R^2(p\cdot q)\\
&\leq 2p^0q^0+2\sqrt{1+R^2|p|^2+R^2|q|^2+R^4|p|^2|q|^2}=4p^0q^0=4v^0u^0.
\end{align*}
For \eqref{elementary 3}, we observe that
\begin{align*}
&(p^0)^2(q^0)^2-(1+R^2(p\cdot q))^2\\
&=1+R^2|p|^2+R^2|q|^2+R^4|p|^2|q|^2-1-2R^2(p\cdot q)-R^4(p\cdot q)^2\\
&\geq R^2|p-q|^2,
\end{align*}
and this implies
\begin{align*}
g^2&=2p^0q^0-2\Big(1+R^2(p\cdot q)\Big)=2\frac{(p^0)^2(q^0)^2
-(1+R^2(p\cdot q))^2}{p^0q^0+1+R^2(p\cdot q)}\\
&\geq\frac{R^2|p-q|^2}{p^0q^0}=\frac{1}{R^2}\frac{|v-u|^2}{v^0u^0}.
\end{align*}
The second inequality of \eqref{elementary 3} comes from \eqref{elementary 5},
and this proves \eqref{elementary 3}. We have assumed that $R=R(t)$ is an increasing
function with $R(0)=1$, hence \eqref{elementary 4} is clear.
For \eqref{elementary 5}, we first note that
\[
R^2(v^0)^2-R^2(u^0)^2=|v|^2-|u|^2=(v-u)\cdot(v+u).
\]
Then, by a direct calculation we have
\begin{align*}
R^2g^2
&=-R^2(p^0-q^0)^2+R^4|p-q|^2
=|v-u|^2-R^2(v^0-u^0)^2\\
&=|v-u|^2-\left(\frac{(v-u)\cdot(v+u)}{R(v^0+u^0)}\right)^2\\
&=|v-u|^2\left(1-\frac{|v+u|^2\cos^2\theta_0}{R^2(v^0+u^0)^2}\right),
\end{align*}
where $\theta_0$ is the angle between $v+u$ and $v-u$.
This proves \eqref{elementary 5}. For \eqref{elementary 6},
we note that
\begin{align*}
s&=(p^0+q^0)^2-R^2|p+q|^2
=(v^0+u^0)^2-R^{-2}|v+u|^2\\
&=(n^0)^2-R^{-2}|n|^2
\leq (n^0)^2-R^{-2}(n\cdot{\omega})^2,
\end{align*}
and this proves \eqref{elementary 6}. For the last inequality \eqref{elementary 7},
we note that
\begin{align*}
s&=(n^0)^2-R^{-2}|n|^2
=(v^0)^2+2v^0u^0+(u^0)^2-R^{-2}|v|^2-2R^{-2}(v\cdot u)-R^{-2}|u|^2\\
&\geq 2+2\sqrt{1+R^{-2}|v|^2}\sqrt{1+R^{-2}|u|^2}-2R^{-2}|v||u|\\
&=2+2\frac{1+R^{-2}|v|^2+R^{-2}|u|^2}{\sqrt{1+R^{-2}|v|^2}\sqrt{1+R^{-2}|u|^2}
+R^{-2}|v||u|}\\
&\geq 2+\frac{1+R^{-2}|v|^2+R^{-2}|u|^2}{v^0u^0}
\geq\frac{v^0}{u^0}+\frac{u^0}{v^0},
\end{align*}
and this completes the proof of the lemma.
\end{proof}

\begin{lemma}\label{Lem derivatives}
The following estimates hold for the quantities defined in the previous sections:
\begin{gather}
\partial_{v^i}v^0=\frac{v^i}{R^2v^0},\quad
|\partial_{v^i}v^0|\leq \frac{1}{R},\label{derivative 1}\\
\partial_{v^i}g=\frac{u^0}{Rg}\left(\frac{v^i}{Rv^0}-\frac{u^i}{Ru^0}\right),\quad
|\partial_{v^i}g|\leq \frac{2u^0}{Rg},\label{derivative 2}\\
\partial_{v^i}\sqrt{s}=\frac{u^0}{R\sqrt{s}}\left(\frac{v^i}{Rv^0}-\frac{u^i}{Ru^0}\right),\quad
|\partial_{v^i}\sqrt{s}|\leq \frac{2u^0}{R\sqrt{s}},\label{derivative 3}\\
\partial_{v^i}\sqrt{(n^0)^2-R^{-2}(n\cdot\omega)^2}
=\frac{u^0}{R\sqrt{(n^0)^2-R^{-2}(n\cdot\omega)^2}}\left(
\frac{v^i}{Rv^0}-\frac{(u\cdot\omega)\omega^i}{Ru^0}\right)\label{derivative 4}\\
\hspace{5cm}+\frac{v^0}{R\sqrt{(n^0)^2-R^{-2}(n\cdot\omega)^2}}\left(
\frac{v^i}{Rv^0}-\frac{(v\cdot\omega)\omega^i}{Rv^0}\right),\cr
|\partial_{v^i}\sqrt{(n^0)^2-R^{-2}(n\cdot\omega)^2}|\leq
\frac{2u^0+2v^0}{R\sqrt{(n^0)^2-R^{-2}(n\cdot\omega)^2}},\label{derivative 5}\\
|\partial_{v^i}g|\leq\frac{u^0\sqrt{v^0u^0}}{R},\quad
|\partial_{v^i}\sqrt{s}|\leq\frac{u^0\sqrt{v^0u^0}}{R},\label{derivative 6}\\
\partial_{v^i}\left[\frac{(n\cdot\omega)n^k}{|n|^2}\right]
=\frac{\omega^in^k}{|n|^2}+\frac{(n\cdot\omega)\delta^{ik}}{|n|^2}
-2\frac{(n\cdot\omega)n^in^k}{|n|^4},\quad
\left|\partial_{v^i}\left[\frac{(n\cdot\omega)n^k}{|n|^2}\right]\right|
\leq\frac{3}{|n|}.\label{derivative 7}
\end{gather}
Note that \eqref{derivative 4} reduces to \eqref{derivative 3}
in the case of $(n\cdot\omega)^2=|n|^2$.
\end{lemma}
\begin{proof}
The first identities of \eqref{derivative 1}--\eqref{derivative 4}
are direct calculations from the definitions of $v^0$, $g$, and $s$.
The inequality of \eqref{derivative 1} is clear. The inequalities
of \eqref{derivative 2}, \eqref{derivative 3}, and \eqref{derivative 5}
come from \eqref{derivative 1}. For the inequalities \eqref{derivative 6}, we note that
\[
\left|\frac{R^{-1}v^i}{\sqrt{1+R^{-2}|v|^2}}-\frac{R^{-1}u^i}{\sqrt{1+R^{-2}|u|^2}}\right|
\leq R^{-1}|v-u|\leq g\sqrt{v^0u^0},
\]
and applying it to \eqref{derivative 2} and \eqref{derivative 3} we obtain
\eqref{derivative 6}. The estimate \eqref{derivative 7} is an elementary calculation,
and this completes the proof of the lemma.
\end{proof}
\begin{lemma}\label{Lem integral}
The following integral for $0\leq\alpha<3$ is estimated as
\[
\int_{\bbr^3}|v-u|^{-\alpha}e^{-|u|^2}du\leq C_\alpha(1+|v|^2)^{-\frac{\alpha}{2}},
\]
where $C_\alpha$ is a positive constant depending on $\alpha$.
\end{lemma}
\begin{proof}
The quantity $|v-u|^{-\alpha}$ is integrable near $v\approx u$ for $\alpha<3$,
hence the above integral
is bounded for small $v$. Then, for large $v$ we separate the domain as
\begin{align*}
\int_{\bbr^3}\frac{1}{|v-u|^\alpha}e^{-|u|^2}du
&\leq\int_{\left\{|v-u|\geq\frac{|v|}{2}\right\}}\cdots\,du
+\int_{\left\{|v-u|\leq\frac{|v|}{2}\right\}}\cdots\,du,
\end{align*}
and note that $|u|\geq\frac{1}{2}|v|$ in the second case.
We estimate them as follows:
\begin{align*}
\int_{\bbr^3}\frac{1}{|v-u|^\alpha}e^{-|u|^2}du
&\leq C_\alpha |v|^{-\alpha}
+e^{-\frac{1}{4}|v|^2}\int_{\left\{|v-u|\leq\frac{|v|}{2}\right\}}
\frac{1}{|v-u|^\alpha}\,du\\
&\leq C_\alpha |v|^{-\alpha}
+C_\alpha |v|^{3-\alpha} e^{-\frac{1}{4}|v|^2}
\leq C_\alpha |v|^{-\alpha}.
\end{align*}
Combining this estimate with the case of small $v$, we obtain the desired result.
\end{proof}

\section{Global existence}\label{Sec global existence}\setcounter{equation}{0}
In this section, we prove global existence of classical solutions to the
Boltzmann equation \eqref{boltzmann}. We first make some remark about the weight function
and then establish several pointwise estimates.
The global existence of classical solutions will be proved in Section
\ref{Sec global existence proof}, and some discussions on the result will be given in
Section \ref{Sec discussions}.

\subsection{Weight function}\label{Sec weight function}
In this part we make some remark about the weight function that we have
chosen in \eqref{norm}. Note that
the energy conservation between colliding particles
can be written as
\begin{equation}\label{energy conserv}
v'^0+u'^0=v^0+u^0,
\end{equation}
and one may try to estimate the equation \eqref{boltzmann} as
\begin{align*}
e^{v^0}\partial_tf&=R^{-3}\iint\Big(\cdots\Big)
e^{-u^0}\Big(e^{v'^0}f(v')e^{u'^0}f(u')-e^{v^0}f(v)e^{u^0}f(u)\Big)\,d\omega\,du\\
&\leq CR^{-3}\|f(t)\|^2\int e^{-u^0}du.
\end{align*}
However, the integral in the last inequality gives a factor $R^3$ because
$u^0$ is defined by \eqref{v},
and then the right hand side will only be bounded by $\|f(t)\|^2$.
Hence, it is not easy to apply the argument of Illner and Shinbrot \cite{IS84},
because integrability is not guaranteed.
Instead, if we use the weight function $e^{|v|^2}$, then the loss term is easily estimated as
\begin{align*}
e^{|v|^2}\partial_tf&=\mbox{(gain term)}
-R^{-3}\iint\Big(\cdots\Big)
e^{-|u|^2}\Big(e^{|v|^2}f(v)e^{|u|^2}f(u)\Big)\,d\omega\,du\\
&\leq \mbox{(gain term)}+CR^{-3}\|f(t)\|^2\int e^{-|u|^2}du,
\end{align*}
and we can apply the method of \cite{IS84} in the case that $R$
grows fast enough such that
$R^{-3}$ is integrable. On the other hand,
as for the gain term,
the energy conservation \eqref{energy conserv} does not apply,
thus we need the restriction on
the scattering kernel given in \eqref{scattering 2}.
Detailed calculations will be given in the following sections.

Another motivation for the weight function is that the representation of post-collisional
momenta \eqref{v' R} or \eqref{v' RS}
converges to that of the Newtonian case as $t$ tends to infinity.
From the definitions of $v^0$ and $u^0$ we observe that $n^0$ and $\sqrt{s}$
tend to $2$ as $t\to\infty$, and Lemma \ref{Lem elementary} shows that
for each $v$ and $u$ we have
\[
v'^k\to\frac{v^k+u^k}{2}+\frac{|v-u|}{2}\omega^k\quad\mbox{as}\quad t\to\infty.
\]
This implies that we have formally a new invariant:
\begin{equation}\label{invariant at infty}
|v'|^2+|u'|^2=|v|^2+|u|^2\quad\mbox{at}\quad t=\infty.
\end{equation}
On the other hand, since $R(0)=1$, we have the following:
\begin{equation}\label{invariant at 0}
\sqrt{1+|v'|^2}+\sqrt{1+|u'|^2}=\sqrt{1+|v|^2}+\sqrt{1+|u|^2}\quad\mbox{at}\quad t=0.
\end{equation}
We may now compare \eqref{invariant at infty} and \eqref{invariant at 0}
to conjecture that the energy conservation \eqref{energy conserv}
has some structure which becomes close to
the Newtonian case rather than the Minkowski case after a large time.
In this sense we choose the weight function as $e^{|v|^2}$.
Since \eqref{invariant at infty} does not apply at a finite time,
we need the restriction \eqref{scattering 2} on the scattering kernel,
and the following lemma shows that the restriction \eqref{scattering 2} is enough to control
the gain term.

\begin{lemma}\label{Lem energy conservation bound}
Let $v$ and $u$ be given.
Suppose that $v'$ and $u'$ are parametrized by \eqref{v' R}
or \eqref{v' RS} with a unit vector $\omega$ or $\hat{\omega}$
on $\bbs^2_R$.
Then, we have the following estimate:
\begin{equation}\label{energy conservation bound}
|v|^2+|u|^2-|v'|^2-|u'|^2\leq B,
\end{equation}
where $B$ is the constant given in \eqref{scattering 2}.
\end{lemma}
\begin{proof}
We first consider the case of \eqref{v' R}.
The quadratic quantity in \eqref{energy conservation bound} can be written as
\begin{align}\label{quadratic quantity}
|v|^2+|u|^2-|v'|^2-|u'|^2=2R^2(v'^0u'^0-v^0u^0),
\end{align}
where we used the energy conservation $v'^0+u'^0=v^0+u^0$.
We apply \eqref{v' R} to get
\begin{align*}
v'^0u'^0=\frac{(v^0)^2}{4}+\frac{(u^0)^2}{4}+\frac{v^0u^0}{2}
-\frac{g^2(n\cdot\omega)^2}{4R^2((n^0)^2-R^{-2}(n\cdot\omega)^2)},
\end{align*}
and then
\begin{align*}
2R^2(v'^0u'^0-v^0u^0)
&=2R^2\left(\frac{(v^0-u^0)^2}{4}
-\frac{g^2(n\cdot\omega)^2}{4R^2((n^0)^2-R^{-2}(n\cdot\omega)^2)}\right)\\
&=\frac{1}{2}\left(Rv^0-Ru^0\right)^2
-\frac{(Rg)^2(n\cdot\omega)^2}{2((Rn^0)^2-(n\cdot\omega)^2)}.
\end{align*}
The first quantity above is written as
\[
\frac{1}{2}\left(Rv^0-Ru^0\right)^2
=\frac{1}{2}\left(\frac{|v|^2-|u|^2}{Rv^0+Ru^0}\right)^2
=\frac{1}{2}\frac{((v+u)\cdot(v-u))^2}{R^2(n^0)^2}
=\frac{|v-u|^2}{2}\frac{|n|^2\cos^2\theta_0}{R^2(n^0)^2},
\]
and the second quantity is written by
\begin{align*}
\frac{(Rg)^2(n\cdot\omega)^2}{2((Rn^0)^2-(n\cdot\omega)^2)}
=\frac{|v-u|^2}{2}\left(1-\frac{|n|^2\cos^2\theta_0}{R^2(n^0)^2}\right)
\frac{(n\cdot\omega)^2}{(Rn^0)^2-(n\cdot\omega)^2},
\end{align*}
where we used \eqref{elementary 5}. Then, \eqref{quadratic quantity} is
estimated as
\begin{align*}
&|v|^2+|u|^2-|v'|^2-|u'|^2\\
&=\frac{|v-u|^2}{2}\left(\frac{|n|^2\cos^2\theta_0}{R^2(n^0)^2}
\left(1+\frac{(n\cdot\omega)^2}{(Rn^0)^2-(n\cdot\omega)^2}\right)
-\frac{(n\cdot\omega)^2}{(Rn^0)^2-(n\cdot\omega)^2}\right)\\
&=\frac{|v-u|^2}{2}\left(\frac{|n|^2\cos^2\theta_0}
{(Rn^0)^2-(n\cdot\omega)^2}
-\frac{(n\cdot\omega)^2}{(Rn^0)^2-(n\cdot\omega)^2}\right)
\leq\frac{|v-u|^2}{2}\frac{|n\times\omega|^2}{R^2s}\leq B,
\end{align*}
where we used the assumption \eqref{scattering 2}.
In the second case, we apply \eqref{v' RS} to the quadratic quantity
in \eqref{quadratic quantity} to obtain
\begin{align*}
2R^2(v'^0u'^0-v^0u^0)
&=2R^2\left(\frac{(v^0-u^0)^2}{4}-\frac{g^2(n\cdot\hat{\omega})^2}{4R^2s}\right)\\
&=\frac{1}{2}(Rv^0-Ru^0)^2-\frac{(Rg)^2(n\cdot\hat{\omega})^2}{2R^2s}.
\end{align*}
The first quantity is the same as above,
\[
\frac{1}{2}\left(Rv^0-Ru^0\right)^2
=\frac{|v-u|^2}{2}\frac{|n|^2\cos^2\theta_0}{R^2(n^0)^2},
\]
while the second quantity is estimated as
\begin{align*}
\frac{(Rg)^2(n\cdot\hat{\omega})^2}{2R^2s}
&=\frac{|v-u|^2}{2}\left(1-\frac{|n|^2\cos^2\theta_0}{R^2(n^0)^2}\right)
\frac{(n\cdot\hat{\omega})^2}{R^2s}\\
&=\frac{|v-u|^2}{2}\frac{(R^2(n^0)^2-|n|^2\cos^2\theta_0)}{R^2(n^0)^2}
\frac{(n\cdot\hat{\omega})^2}{R^2s}
\geq\frac{|v-u|^2}{2}\frac{(n\cdot\hat{\omega})^2}{R^2(n^0)^2}.
\end{align*}
Then, \eqref{quadratic quantity} is estimated as
\begin{align*}
|v|^2+|u|^2-|v'|^2-|u'|^2
&\leq\frac{|v-u|^2}{2}
\frac{(|n|^2\cos^2\theta_0-(n\cdot\hat{\omega})^2)}{R^2(n^0)^2}\\
&\leq\frac{|v-u|^2}{2}
\frac{|n\times\hat{\omega}|^2}{R^2(n^0)^2}
\leq\frac{|v-u|^2}{2}
\frac{|n\times\hat{\omega}|^2}{R^2s}\leq B.
\end{align*}
This completes the proof of the lemma.
\end{proof}

\subsection{Pointwise estimates}\label{Sec pointwise estimates}
In this part, we collect several pointwise estimates which will be used
in the following section for the global existence theorem.
The following lemma is trivial from the assumption on the angular part
of the scattering kernel.
\begin{lemma}\label{Lem integral trivial}
Suppose that $\sigma_0(\omega)$ satisfies \eqref{scattering 2}. Then, we have
\[
\iint\sigma_0(\omega)e^{-|u|^2}d\omega\,du\leq C.
\]
\end{lemma}
\begin{proof}
Since $\sigma_0(\omega)$ is bounded, the lemma is clear.
\end{proof}
\begin{lemma}\label{Lem integral 2}
For $0\leq \beta<4$, we have the following estimate:
\[
\int_{\bbr^3}v_\phi g^{-\beta}e^{-|u|^2}du\leq
\left\{
\begin{aligned}
&C\quad&\mbox{for}\quad 0\leq \beta\leq 1,\\
&CR^{\beta-1}\quad&\mbox{for}\quad 1\leq \beta<4,
\end{aligned}
\right.
\]
where $C$ is a positive constant depending on $\beta$.
\end{lemma}
\begin{proof}
Since the lower and upper bounds for $g$ are different \eqref{elementary 3},
we separate the cases as follows. Note that $R(t)$ is an increasing function.\bigskip

\noindent Case 1. ($0\leq \beta\leq 1$) By the definition of $v_\phi$
and \eqref{elementary 2}, we obtain
\begin{align*}
\int v_\phi g^{-\beta}e^{-|u|^2}du
&=\int \frac{g^{1-\beta}\sqrt{s}}{v^0u^0}e^{-|u|^2}du
\leq C\int \frac{(v^0u^0)^{\frac{1-\beta}{2}}(v^0u^0)^{\frac{1}{2}}}{v^0u^0}e^{-|u|^2}du\\
&\leq C\int (v^0u^0)^{-\frac{\beta}{2}}e^{-|u|^2}du\leq C.
\end{align*}

\noindent Case 2. ($1\leq \beta\leq 2$)
We use the lower bound for $g$ in \eqref{elementary 3}
with Lemma \ref{Lem integral} to estimate
\begin{align*}
\int v_\phi g^{-\beta}e^{-|u|^2}du
&=\int \frac{\sqrt{s}}{v^0u^0}\frac{1}{g^{\beta-1}}e^{-|u|^2}du
\leq C\int\frac{1}{\sqrt{v^0u^0}}
\frac{R^{\beta-1}(v^0u^0)^{\frac{\beta-1}{2}}}{|v-u|^{\beta-1}}e^{-|u|^2}du\\
&\leq CR^{\beta-1}\int\frac{1}{|v-u|^{\beta-1}}
\frac{1}{(v^0u^0)^{\frac{2-\beta}{2}}}e^{-|u|^2}du\\
&\leq CR^{\beta-1}\int\frac{1}{|v-u|^{\beta-1}}e^{-|u|^2}du
\leq C_\beta R^{\beta-1}(1+|v|^2)^{-\frac{\beta-1}{2}},
\end{align*}
where we used Lemma \ref{Lem integral}.

\noindent Case 3. ($2\leq \beta<4$)
We use \eqref{elementary 4} and Lemma \ref{Lem integral} to estimate
\begin{align*}
\int v_\phi g^{-\beta}e^{-|u|^2}du
&=\int \frac{\sqrt{s}}{v^0u^0}\frac{1}{g^{\beta-1}}e^{-|u|^2}du
\leq C\int\frac{1}{\sqrt{v^0u^0}}
\frac{R^{\beta-1}(v^0u^0)^{\frac{\beta-1}{2}}}{|v-u|^{\beta-1}}e^{-|u|^2}du\\
&\leq CR^{\beta-1}\int\frac{1}{|v-u|^{\beta-1}}(v^0u^0)^{\frac{\beta-2}{2}}e^{-|u|^2}du\\
&\leq CR^{\beta-1}\int\frac{(1+|v|^2)^{\frac{\beta-2}{4}}
(1+|u|^2)^{\frac{\beta-2}{4}}}{|v-u|^{\beta-1}}
e^{-|u|^2}du\\
&\leq C_\beta R^{\beta-1}(1+|v|^2)^{\frac{\beta-2}{4}-\frac{\beta-1}{2}}
\leq C_\beta R^{\beta-1}(1+|v|^2)^{-\frac{\beta}{4}}.
\end{align*}
We combine the above estimates to complete the proof of the lemma.
\end{proof}

\begin{lemma}\label{Lem control derivatives GS}
Consider the representation for $v'$ in \eqref{v' R}. We have the following estimate:
\[
|\partial_{v^i}{v'^k}|\leq Cv^0(u^0)^4,
\]
where the constant $C$ does not depend on $R$.
\end{lemma}
\begin{proof}
For simplicity, let $r$ denote
\[
r=\sqrt{(n^0)^2-R^{-2}(n\cdot\omega)^2}.
\]
Then, we take $v^i$-derivative on $v'^k$ to have
\begin{align*}
\partial_{v^i}{v'^k}&=\frac{\delta^{ik}}{2}
+\partial_{v^i}\left[\frac{Rg}{2}\frac{n^0\omega^k}{\sqrt{(n^0)^2
-R^{-2}(n\cdot\omega)^2}}\right]\\
&=\frac{\delta^{ik}}{2}
+\frac{R(\partial_{v^i}g)}{2}\frac{n^0\omega^k}{r}
+\frac{Rg}{2}\frac{(\partial_{v^i}v^0)\omega^k}{r}
-\frac{Rg}{2}\frac{n^0\omega^k}{r^2}(\partial_{v^i}r),
\end{align*}
where $\delta^{ik}$ is the Kronecker delta.
We now collect \eqref{derivative 6},
\eqref{elementary 6}, and \eqref{elementary 7} to estimate
\begin{align*}
\left|\frac{R(\partial_{v^i}g)}{2}\frac{n^0\omega^k}{r}\right|
&\leq\frac{R}{2}\frac{u^0\sqrt{v^0u^0}}{R}(v^0+u^0)
\sqrt{\frac{u^0}{v^0}}\leq C\Big(v^0(u^0)^2+(u^0)^3\Big),
\end{align*}
and \eqref{elementary 2}, \eqref{derivative 1}, \eqref{elementary 6},
and \eqref{elementary 7} to estimate
\begin{align*}
\left|\frac{Rg}{2}\frac{(\partial_{v^i}v^0)\omega^k}{r}\right|
\leq \frac{R}{2}2\sqrt{v^0u^0}\frac{1}{R}\sqrt{\frac{u^0}{v^0}}
\leq Cu^0,
\end{align*}
and \eqref{elementary 2}, \eqref{elementary 6}, \eqref{elementary 7},
and \eqref{derivative 5} to estiamte
\begin{align*}
\left|\frac{Rg}{2}\frac{n^0\omega^k}{r^2}(\partial_{v^i}r)\right|
\leq\frac{R}{2}2\sqrt{v^0u^0}(v^0+u^0)\left(\frac{u^0}{v^0}\right)
\frac{2(v^0+u^0)}{R}\sqrt{\frac{u^0}{v^0}}
\leq C\left(v^0(u^0)^2+(u^0)^3+\frac{(u^0)^4}{v^0}\right).
\end{align*}
Since $v^0\geq 1$ and $u^0\geq 1$, we combine the above estimates to
obtain the desired result.
\end{proof}

\begin{lemma}\label{Lem control derivatives S}
Consider the representation for $v'$ in \eqref{v' RS}. We have the following estimate:
\[
|\partial_{v^i}{v'^k}|\leq
C\left(\frac{Rv^0}{|v-u|}+\frac{Rv^0}{|v+u|}+\frac{R^2(v^0)^2}{|v-u|^2}\right)(u^0)^3,
\]
where the constant $C$ does not depend on $R$.
\end{lemma}
\begin{proof}
The proof of this lemma is also a direct calculation
as in Lemma \ref{Lem control derivatives GS}.
We take $v^i$-derivative on $v'^k$ in the representation of \eqref{v' RS} to get
\begin{align*}
\partial_{v^i}{v'^k}&=\frac{\delta^{ik}}{2}
+\partial_{v^i}\left[\frac{Rg}{2}
\left(\hat{\omega}^k-\frac{(n\cdot\hat{\omega})n^k}{|n|^2}
+\frac{n^0}{\sqrt{s}}\frac{(n\cdot\hat{\omega})n^k}{|n|^2}\right)\right]\\
&=\frac{\delta^{ik}}{2}
+\frac{R}{2}(\partial_{v^i}g)\hat{\omega}^k
-\frac{R}{2}(\partial_{v^i}g)\frac{(n\cdot\hat{\omega})n^k}{|n|^2}
-\frac{Rg}{2}\partial_{v^i}\left[\frac{(n\cdot\hat{\omega})n^k}{|n|^2}\right]
+\frac{R}{2}(\partial_{v^i}g)\frac{n^0}{\sqrt{s}}\frac{(n\cdot\hat{\omega})n^k}{|n|^2}\\
&\quad+\frac{Rg}{2}\frac{(\partial_{v^i}v^0)}{\sqrt{s}}\frac{(n\cdot\hat{\omega})n^k}{|n|^2}
-\frac{Rg}{2}\frac{n^0}{s}(\partial_{v^i}\sqrt{s})\frac{(n\cdot\hat{\omega})n^k}{|n|^2}
+\frac{Rg}{2}\frac{n^0}{\sqrt{s}}\partial_{v^i}\left[\frac{(n\cdot\hat{\omega})n^k}{|n|^2}\right].
\end{align*}
As in the proof of Lemma \ref{Lem control derivatives GS}, we separately estimate
the above quantities. We first apply \eqref{derivative 2} and \eqref{elementary 3}
to the second quantity above to get
\[
\left|\frac{R}{2}(\partial_{v^i}g)\hat{\omega}^k\right|
\leq\frac{R}{2}\frac{2u^0}{Rg}\leq\frac{u^0}{g}\leq\frac{Ru^0\sqrt{v^0u^0}}{|v-u|},
\]
and similarly
\[
\left|\frac{R}{2}(\partial_{v^i}g)\frac{(n\cdot\hat{\omega})n^k}{|n|^2}\right|
\leq\frac{Ru^0\sqrt{v^0u^0}}{|v-u|}.
\]
The inequalities in \eqref{derivative 7} and \eqref{elementary 2} give
\[
\left|\frac{Rg}{2}\partial_{v^i}\left[\frac{(n\cdot\hat{\omega})n^k}{|n|^2}\right]\right|
\leq C\frac{Rg}{|n|}\leq C\frac{R\sqrt{v^0u^0}}{|v+u|}.
\]
For the fifth quantity, we apply \eqref{derivative 2}, \eqref{elementary 2},
and \eqref{elementary 3} to get
\[
\left|\frac{R}{2}(\partial_{v^i}g)\frac{n^0}{\sqrt{s}}
\frac{(n\cdot\hat{\omega})n^k}{|n|^2}\right|
\leq\frac{R}{2}\frac{2u^0}{Rg}\frac{(v^0+u^0)}{\sqrt{s}}
\leq C\frac{R^2v^0(u^0)^2(v^0+u^0)}{|v-u|^2}.
\]
Applying \eqref{elementary 2}, \eqref{elementary 3}, and \eqref{derivative 1}, we have
\[
\left|\frac{Rg}{2}\frac{(\partial_{v^i}v^0)}{\sqrt{s}}
\frac{(n\cdot\hat{\omega})n^k}{|n|^2}\right|
\leq\frac{R}{2}2\sqrt{v^0u^0}\frac{1}{R}\frac{R\sqrt{v^0u^0}}{|v-u|}
\leq \frac{Rv^0u^0}{|v-u|},
\]
and by \eqref{elementary 1}, \eqref{derivative 3}, \eqref{elementary 2},
and \eqref{elementary 3}, we obtain
\[
\left|\frac{Rg}{2}\frac{n^0}{s}(\partial_{v^i}\sqrt{s})
\frac{(n\cdot\hat{\omega})n^k}{|n|^2}\right|
\leq C\frac{u^0(v^0+u^0)}{s}
\leq C\frac{R^2v^0(u^0)^2(v^0+u^0)}{|v-u|^2}.
\]
The last quantity is estimated by \eqref{elementary 1} and \eqref{derivative 7} as
\[
\left|\frac{Rg}{2}\frac{n^0}{\sqrt{s}}\partial_{v^i}
\left[\frac{(n\cdot\hat{\omega})n^k}{|n|^2}\right]\right|
\leq C\frac{R(v^0+u^0)}{|v+u|}.
\]
Since $v^0\geq 1$ and $u^0\geq 1$, we combine the above estimates to obtain the desired
result.
\end{proof}
The above estimates in Lemma \ref{Lem control derivatives GS} and \ref{Lem control derivatives S}
will be crucially used in the proof of existence of classical solutions.
To estimate the collision operator, we will decompose the integration domain into three different
cases. We fix a finite time $t$ and separate the cases as follows:
\[
\begin{aligned}
&\text{(Case 1)}\qquad|v|\leq R,\\
&\text{(Case 2)}\qquad|v|\geq R\quad\text{and}\quad |v|\leq 2|u|,\\
&\text{(Case 3)}\qquad|v|\geq R\quad\text{and}\quad |v|\geq 2|u|.
\end{aligned}
\]
In the first and second cases, the estimate of Lemma \ref{Lem control derivatives GS}
is further estimated as follows:
\begin{align}\label{case 1}
&\text{(Case 1)}\quad\Longrightarrow\quad
|\partial_{v^i}v'^k|\leq C\sqrt{1+R^{-2}|v|^2}(u^0)^4
\leq C(u^0)^4,\\
&\text{(Case 2)}\quad\Longrightarrow\quad
|\partial_{v^i}v'^k|\leq C\sqrt{1+R^{-2}|v|^2}(u^0)^4
\leq C(u^0)^5.\label{case 2}\\
\intertext{In the third case, we note that $|v\pm u|\geq\frac{1}{2}|v|$ and
use Lemma \ref{Lem control derivatives S} to estimate}
&\text{(Case 3)}\quad\Longrightarrow\quad|\partial_{v^i}v'^k|\leq C\left(\frac{Rv^0}{|v|}
+\frac{(Rv^0)^2}{|v|^2}\right)(u^0)^3\leq C(u^0)^3.\label{case 3}
\end{align}
We observe that the estimates \eqref{case 1}--\eqref{case 3}
do not produce a growth in $v$. The main idea is to use the representation \eqref{v' RS}
together with \eqref{v' R} and decompose the integration domain into the three cases as above.
This argument was originally suggested by Guo and Strain
in \cite{GS12}, and the existence of classical solutions
of the relativistic Vlasov-Maxwell-Boltzmann system was proved.
In the following section, we will use the estimates \eqref{case 1}--\eqref{case 3}
to prove the existence of classical solutions to the Boltzmann equation
in the RW spacetime.

\subsection{Global existence}\label{Sec global existence proof}
We now prove the global existence theorem.

\begin{lemma}\label{Lem estimate 123}
We have the following estimate for $f$:
\[
|e^{|v|^2}f(t,v)|\leq \|f_0\|+C\|f(t)\|^2\int_0^tR^{-3}(s)+R^{b-4}(s)\,ds,
\]
where $C$ is a positive constant depending on $b$.
\end{lemma}
\begin{proof}
Multiplying the weight function to the Boltzmann equation \eqref{boltzmann} as in
Section \ref{Sec weight function}, we obtain
\begin{align}
&\partial_t\Big[e^{|v|^2}f(t,v)\Big]\label{estimate 1}\\
&=R^{-3}\iint v_\phi\sigma(g,\omega)
e^{|v|^2+|u|^2-|v'|^2-|u'|^2}\Big(e^{|v'|^2}f(v')e^{|u'|^2}f(u')\Big)e^{-|u|^2}d\omega\,du
\nonumber\\
&\quad-R^{-3}\iint v_\phi\sigma(g,\omega)
\Big(e^{|v|^2}f(v)e^{|u|^2}f(u)\Big)e^{-|u|^2}d\omega\,du\nonumber\\
&=:I_1+I_2,\nonumber
\end{align}
and $I_2$ is easily estimated as
\begin{align}
|I_2|&\leq R^{-3}\|f(t)\|^2
\iint v_\phi\sigma(g,\omega)e^{-|u|^2}d\omega\,du\label{estimate 2}\\
&\leq CR^{-3}\|f(t)\|^2
\left(\int v_\phi e^{-|u|^2}du
+\int v_\phi g^{-b} e^{-|u|^2}du\right)\nonumber\\
&\leq C\Big(R^{-3}(t)+R^{b-4}(t)\Big)\|f(t)\|^2,\nonumber
\end{align}
where we used Lemma \ref{Lem integral 2}.
For $I_1$, we use Lemma \ref{Lem energy conservation bound} to estimate
\begin{align}
|I_1|&\leq R^{-3}\|f(t)\|^2
\iint v_\phi\sigma(g,\omega)
e^{|v|^2+|u|^2-|v'|^2-|u'|^2}e^{-|u|^2}d\omega\,du\label{estimate 3}\\
&\leq CR^{-3}\|f(t)\|^2
\iint v_\phi\sigma(g,\omega)
e^{-|u|^2}d\omega\,du\nonumber\\
&\leq C\Big(R^{-3}(t)+R^{b-4}(t)\Big)\|f(t)\|^2.\nonumber
\end{align}
Then, we obtain a differential inequality from \eqref{estimate 1}--\eqref{estimate 3}.
Integrating it from $0$ to $t$, we have
\[
e^{|v|^2}f(t,v)\leq \|f_0\|+C\|f(t)\|^2\int_0^tR^{-3}(s)+R^{b-4}(s)\,ds,
\]
and this completes the proof.
\end{proof}
\begin{lemma}\label{Lem estimate 45678}
We have the following estimate for $\partial_vf$:
\[
|e^{|v|^2}\partial_{v^i}f(t,v)|\leq
\|f_0\|+C\|f(t)\|^2\int_0^tR^{-3}(s)+R^{b-4}(s)\,ds,\quad i=1,2,3,
\]
where $C$ is a positive constant depending on $b$.
\end{lemma}
\begin{proof}
In this lemma, we make use of \eqref{case 1}--\eqref{case 3}.
We take $v^i$-derivative on the Boltzmann equation \eqref{boltzmann} and multiply
the weight function to have
the following equation for $\partial_{v^i}f$:
\begin{align}\label{estimate 4}
\partial_t\Big[e^{|v|^2}\partial_{v^i}f\Big]
&=R^{-3}\iint\partial_{v^i}\Big[v_\phi\sigma(g,\omega)\Big]e^{|v|^2}
\Big(f(v')f(u')-f(v)f(u)\Big)\,d\omega\,du\\
&\quad+R^{-3}\iint v_\phi\sigma(g,\omega)e^{|v|^2}
\partial_{v^i}\Big[f(v')f(u')\Big]\,d\omega\,du\nonumber\\
&\quad-R^{-3}\iint v_\phi\sigma(g,\omega)e^{|v|^2}
(\partial_{v^i}f)(v)f(u)\,d\omega\,du\nonumber\\
&=:J_1+J_2+J_3.\nonumber
\end{align}
For $J_1$, we note that
\begin{align*}
\partial_{v^i}\Big[v_\phi\sigma(g,\omega)\Big]
&=(\partial_{v^i}g)\frac{\sqrt{s}}{v^0u^0}\sigma(g,\omega)
+(\partial_{v^i}\sqrt{s})\frac{g}{v^0u^0}\sigma(g,\omega)\\
&\quad-(\partial_{v^i}v^0)\frac{g\sqrt{s}}{(v^0)^2u^0}\sigma(g,\omega)
+\frac{g\sqrt{s}}{v^0u^0}(\partial_{v^i}g)\partial_g\sigma(g,\omega).
\end{align*}
Applying \eqref{derivative 6}, \eqref{derivative 1}, and \eqref{elementary 2} to the above,
we estimate
\begin{align*}
\Big|\partial_{v^i}\Big[v_\phi\sigma(g,\omega)\Big]\Big|
&\leq \frac{u^0\sqrt{v^0u^0}}{R}\frac{\sqrt{s}}{v^0u^0}\sigma(g,\omega)
+\frac{u^0\sqrt{v^0u^0}}{R}\frac{g}{v^0u^0}\sigma(g,\omega)\\
&\quad+\frac{1}{R}\frac{g\sqrt{s}}{(v^0)^2u^0}\sigma(g,\omega)
+\frac{u^0\sqrt{v^0u^0}}{R}\frac{g\sqrt{s}}{v^0u^0}
|\partial_g\sigma(g,\omega)|\\
&\leq \frac{Cu^0}{R}\Big(\sigma(g,\omega)
+g|\partial_g\sigma(g,\omega)|\Big).
\end{align*}
By the assumption on the scattering kernel \eqref{scattering 1}, we get
\begin{align*}
\Big|\partial_{v^i}\Big[v_\phi\sigma(g,\omega)\Big]\Big|
\leq CR^{-1}u^0(1+g^{-b})\sigma_0(\omega).
\end{align*}
Since $u^0\leq\sqrt{1+|u|^2}$, we estimate $J_1$ by the same argument as in Lemma
\ref{Lem estimate 123}:
\begin{align}\label{estimate 5}
|J_1|\leq C\Big(R^{-4}(t)+R^{b-5}(t)\Big)\|f(t)\|^2.
\end{align}
The estimate of $J_3$ is exactly the same with Lemma \ref{Lem estimate 123}.
We obtain
\begin{align}\label{estimate 6}
|J_3|\leq C\Big(R^{-3}(t)+R^{b-4}(t)\Big)\|f(t)\|^2.
\end{align}
For $J_2$, we write
\begin{align*}
J_2=R^{-3}\iint v_{\phi}\sigma(g,\omega)e^{|v|^2}
\Big((\partial_{v^i}v')(\partial_vf)(v')f(u')
+f(v')(\partial_{v^i}u')(\partial_vf)(u')\Big)\,d\omega\,du,
\end{align*}
and separate the cases as in \eqref{case 1}--\eqref{case 3}.
We fix a momentum $v$ and note that
$R(t)$ is an increasing function with $R(0)=1$.
Then, we can find a finite time $t_0$ such that
\[
t\geq t_0\quad\Longleftrightarrow\quad |v|\leq R(t).
\]
Note that $t_0$ can be zero for small $v$.
We first consider the case of $t\geq t_0$. In this case, we apply \eqref{case 1}, and
the quantities $\partial_{v^i}v'$ and $\partial_{v^i}u'$ are bounded by $(u^0)^4$.
This quantity can be controlled by the weight function as in \eqref{estimate 3},
we obtain the following estimate:
\begin{align}\label{estimate 7}
t\geq t_0\quad\Longrightarrow\quad |J_2|\leq
C\Big(R^{-3}(t)+R^{b-4}(t)\Big)\|f(t)\|^2.
\end{align}
In the case of $t\leq t_0$, we decompose the integration domain as in \eqref{case 2}
and \eqref{case 3}. We write $J_2$ as
\[
J_2=R^{-3}\iint_{|v|\leq 2|u|}\cdots\, d\omega\,du
+R^{-3}\iint_{|v|\geq 2|u|}\cdots\, d\omega\,du=:J_{21}+J_{22},
\]
and apply \eqref{case 2} to $J_{21}$. Then, the quantities $\partial_{v^i}v'$
and $\partial_{v^i}u'$ are bounded by $(u^0)^5$. In the case of $J_{22}$,
we consider the post-collisional momenta $v'$ and $u'$ in the representation of \eqref{v' RS}.
Applying \eqref{case 3} to $J_{22}$, we control $\partial_{v^i}v'$
and $\partial_{v^i}u'$ to be bounded by $(u^0)^3$. By the same argument as in the
proof of Lemma \ref{Lem estimate 123}, we now obtain the following estimate:
\begin{align}\label{estimate 8}
t\leq t_0\quad\Longrightarrow\quad |J_2|\leq
C\Big(R^{-3}(t)+R^{b-4}(t)\Big)\|f(t)\|^2.
\end{align}
We combine the estimates \eqref{estimate 4}--\eqref{estimate 8}
and apply them to \eqref{estimate 4} to obtain
\[
|e^{|v|^2}\partial_{v^i}f(t,v)|\leq
\|f_0\|+C\|f(t)\|^2\int_0^tR^{-3}(s)+R^{b-4}(s)\,ds,
\]
and this completes the proof.
\end{proof}

We are now ready to prove global existence of solutions to the Boltzmann equation
in the RW spacetime. With Lemma \ref{Lem estimate 123} and \ref{Lem estimate 45678}
it is easy to apply the method of \cite{IS84}. The estimates of those lemmas give
the following estimate for $f$:
\[
\|f(t)\|\leq \|f_0\|+C\|f(t)\|^2\int_0^tR^{-3}(s)+R^{b-4}(s)\,ds
\]
for some positive constant $C$.
In the case that the scale factor $R$ grows fast enough such that the integral above
converges, we obtain the following inequality:
\[
\|f(t)\|\leq \|f_0\|+C\|f(t)\|^2.
\]
The above inequality shows that if initial data is sufficiently small, then global existence
of solutions is guaranteed. For detailed arguments of the proof of this framework, we refer to
\cite{G}. By following the proofs of \cite{G06,G01,IS84,S101}, we finally obtain the following
theorem.

\begin{theorem}\label{Thm}
Consider the relativistic Boltzmann equation in the spatially flat
Robertson-Walker spacetime in the form of \eqref{boltzmann}.
Suppose that the scattering kernel satisfies \eqref{scattering 1} and \eqref{scattering 2},
and let the scale factor $R$
be given and satisfy \eqref{scale factor} together with the following condition:
\begin{equation}\label{integrability condition}
\int_0^\infty R^{-3}(t)+R^{b-4}(t)\,dt<\infty,
\end{equation}
where $b$ is the constant given in \eqref{scattering 1}.
Let $f_0$ be an initial data such that it is differentiable and satisfies
$\|f_0\|<\varepsilon$ for some positive constant $\varepsilon$.
If the constant $\varepsilon$ is sufficiently small, then there exists
a unique nonnegative classical solution
of the Boltzmann equation \eqref{boltzmann} such that
\begin{equation}\label{property of solution}
\sup_{0\leq t<\infty}\|f(t)\|\leq C_\varepsilon,
\end{equation}
where $C_\varepsilon$ is some positive constant depending on $\varepsilon$.
\end{theorem}

\subsection{Discussions}\label{Sec discussions}
In Theorem \ref{Thm}, we have obtained global existence of classical solutions to the Boltzmann
equation \eqref{boltzmann}. Note that the equation \eqref{boltzmann} is written in the transformed
variable $v$, hence we write it back in the original variable $p$,
and then \eqref{property of solution} can be written as
\begin{equation}\label{pointwise estimate}
f(t,v)\leq C_\varepsilon e^{-|v|^2}\quad\text{or}\quad
f(t,p)\leq C_\varepsilon e^{-R^4(t)|p|^2}.
\end{equation}
We may compare this result with the Vlasov case. The Vlasov equation in the RW spacetime
is obtained by simply ignoring the right hand side of the Boltzmann equation \eqref{boltzmann},
i.e. $\partial_tf=0$ in the transformed variable, and we get $f(t,v)=f_0(v)$.
It is usual in the Vlasov case to assume that initial data has a compact support such that
$f_0(v)=0$ for $|v|\geq C$, hence we obtain
\[
f(t,v)=0\quad\text{for}\quad |v|\geq C,\quad\text{or}\quad
f(t,p)=0\quad\text{for}\quad |p|\geq CR^{-2}(t).
\]
This kind of estimates for the Vlasov equation has already been obtained
in more general cases.
In \cite{L04}, the author studied the Einstein-Vlasov system with
a positive cosmological constant in the spacetimes
of all Bianchi types except IX.
She showed that the distribution function $f$ satisfies the above
estimates with $R(t)\approx e^{c t}$ (see Theorem 3.8 of \cite{L04}).
A similar result has been obtained in \cite{N10}, where the Einstein-Vlasov system was studied
with zero cosmological constant.
We also remark that the integrability condition \eqref{integrability condition}
does not seem to be a strong restriction.
The condition \eqref{integrability condition} is indeed satisfied by several special solutions
of the Einstein equations, for instance the Einstein-de Sitter model $R(t)=t^{2/3}$
and the de Sitter spacetime $R(t)=e^{Ht}$
(see \cite{R} for more details). We expect that the condition \eqref{integrability condition}
holds in more general spacetimes.

To summarize, we have studied the relativistic Boltzmann equation in a given RW spacetime.
By applying the argument of Guo and Strain \cite{GS12}, we could prove the global existence of
classical solutions together with the asymptotic behaviour \eqref{pointwise estimate}.
In this paper, we assumed that the scale factor $R$ is given, but the result of this paper
can be easily applied to the Einstein-Boltzmann case.
The Boltzmann equation is coupled to the Einstein equations through the energy-momentum
tensor, and the energy-momentum tensor of the Boltzmann equation
has the same form with that of the Vlasov equation. Hence,
existence of solutions will be proved
in a way similar to the Einstein-Vlasov cases. We also hope that this work can be applied to
more general cases.

\subsection*{Acknowledgements}
This work was supported by a grant from the Kyung Hee University in 2013 (KHU-20130368).


\begin{thebibliography}{99}

\bibitem{A11} Andr{\'e}asson, H.:
The Einstein-Vlasov system/kinetic theory.
{\it Living Rev. Relativity} 14, (2011), 4.
URL: http://www.livingreviews.org/lrr-2011-4

\bibitem{B73} Bancel, D.:
Probl{\`e}me de Cauchy pour l'{\'e}quation de Boltzmann en relativit{\'e} g{\'e}n{\'e}rale.
{\it Ann. Inst. H. Poincar{\'e} Sect. A (N.S.)} 18 (1973), 263--284.

\bibitem{BCB73} Bancel, D., Choquet-Bruhat, Y.:
Existence, uniqueness, and local stability for the Einstein-Maxwell-Boltzmann system.
{\it Comm. Math. Phys.} 33 (1973), 83--96.

\bibitem{CIP} Cercignani, C., Illner, R., Pulvirenti, M.:
The mathematical theory of dilute gases.
Applied Mathematical Sciences, 106.
{\it Springer-Verlag, New York,} 1994.

\bibitem{dvv} de Groot, S. R., van Leeuwen, W. A., van Weert, C. G.:
Relativistic kinetic theory. Principles and applications.
{\it North-Holland Publishing Co., Amsterdam-New York,} 1980.

\bibitem{E} Ehlers, J.: General relativity and kinetic theory.
{\it General relativity and cosmology (Proc. Internat. School of Physics ``Enrico Fermi'', Italian Phys.
Soc., Varenna, 1969),} pp. 1--70. Academic Press, New York, 1971.

\bibitem{G} Glassey, R. T.:
The Cauchy problem in kinetic theory.
{\it Society for Industrial and Applied Mathematics (SIAM), Philadelphia, PA,} 1996.

\bibitem{G06} Glassey, R. T.:
Global solutions to the Cauchy problem for the relativistic Boltzmann equation
with near-vacuum data.
{\it Comm. Math. Phys.} 264 (2006), 705--724.


\bibitem{GS91} Glassey, R. T., Strauss, W. A.: On the derivatives of
the collision map of relativistic particles.
{\it Transport Theory Statist. Phys.} 20 (1991), no. 1, 55--68.

\bibitem{GS93} Glassey, R. T., Strauss, W. A.: Asymptotic stability
of the relativistic Maxwellian.
{\it Publ. Res. Inst. Math. Sci.} 29 (1993), no. 2, 301--347.

\bibitem{G01} Guo, Y.:
The Vlasov-Poisson-Boltzmann system near vacuum.
{\it Comm. Math. Phys.} 218 (2001), 293--313.


\bibitem{GS12} Guo, Y., Strain, R. M.:
Momentum regularity and stability of the relativistic Vlasov-Maxwell-Boltzmann system.
{\it Comm. Math. Phys.} 310 (2012), no. 3, 649--673.

\bibitem{IS84} Illner, R., Shinbrot, M.:
The Boltzmann equation: global existence for a rare gas in an infinite vacuum.
{\it Comm. Math. Phys.} 95 (1984), 217--226.

\bibitem{L04} Lee, H.:
Asymptotic behaviour of the Einstein-Vlasov system with a positive cosmological constant.
{\it Math. Proc. Camb. Phil. Soc.} 137 (2004), 495--509.

\bibitem{L13} Lee, H.:
Global solutions of the Vlasov-Poisson-Boltzmann system in a cosmological setting.
To appear in {\it J. Math. Phys.}

\bibitem{LR13} Lee, H., Rendall, A. D.: The Einstein-Boltzmann system and positivity.
{\it J. Hyperbolic Differ. Equ.} 10 (2013), no. 1, 77--104.

\bibitem{LR131} Lee, H., Rendall, A. D.:
The spatially homogeneous relativistic Boltzmann equation with a hard potential.
Preprint arXiv:1301.0106v1.

\bibitem{ND06} Noutchegueme, N., Dongo, D.:
Global existence of solutions for the Einstein-Boltzmann system in a Bianchi type I
spacetime for arbitrary large initial data.
{\it Classical Quantum Gravity} 23 (2006), no. 9, 2979--3003.

\bibitem{NDT05} Noutchegueme, N., Dongo, D., Takou, E.:
Global existence of solutions for the relativistic
Boltzmann equation with arbitrary large initial data on a Bianchi Type I space-time.
{\it Gen. Relativity Gravitation} 37 (2005), no. 12, 2047--2062.

\bibitem{NT05} Noutchegueme, N., Takou, E.:
Global existence of solutions for the Einstein-Boltzmann system with cosmological
constant in the Robertson-Walker space-time.
{\it Commun. Math. Sci.} 4 (2006), no. 2, 291--314.

\bibitem{N10} Nungesser, E.:
Isotropization of non-diagonal Bianchi I spacetimes with collisionless matter
at late times assuming small data.
{\it Class. Quantum Grav.} 27 (2010), 235025.

\bibitem{R} Rendall, A. D.:
Partial differential equations in General Relativity.
Oxford Graduate Texts in Mathematics, 16.
{\it Oxford University Press, Oxford,} 2008.

\bibitem{R94} Rendall, A. D.:
Cosmic censorship for some spatially homogeneous cosmological models.
{\it Ann. Phys. (N.Y.)} 233 (1994), 82--96.

\bibitem{R95} Rendall, A. D.:
Global properties of locally homogeneous cosmological models with matter.
{\it Math. Proc. Camb. Phil. Soc.} 118 (1995), 511--526.

\bibitem{R96} Rendall, A. D.:
The initial singularity in solutions of the Einstein-Vlasov system of Bianchi type I.
{\it J. Math. Phys.} 37 (1996), 438--451.

\bibitem{R05} Rendall, A. D.:
Theorems on existence and global dynamics for the Einstein equations.
{\it Living Rev. Relativity} 8, (2005), 6. URL:
http://www.livingreviews.org/lrr-2005-6

\bibitem{RT99} Rendall, A. D., Tod, K. P.:
Dynamics of spatially homogeneous solutions of the Einstein-Vlasov equations
which are locally rotationally symmetric.
{\it Class. Quantum Grav.} 16 (1999), 1705--1726.

\bibitem{RU00} Rendall, A. D., Uggla, C.:
Dynamics of spatially homogeneous locally rotationally symmetric solutions
of the Einstein-Vlasov equations.
{\it Class. Quantum Grav.} 17 (2000), 4697--4713.

\bibitem{St} Stewart, J. M.: Non-equilibrium relativistic kinetic theory.
Lecture Notes in Physics, 10. {\it Springer-Verlag,
Berlin-Heidelberg-New York,} 1971.


\bibitem{S101} Strain, R. M.:
Global Newtonian limit for the relativistic Boltzmann equation near vacuum.
{\it SIAM J. Math. Anal.} 42 (2010), no. 4, 1568--1601.

\bibitem{S102} Strain, R. M.:
Asymptotic stability of the relativistic Boltzmann equation for the soft potentials.
{\it Comm. Math. Phys.} 300 (2010), no. 2, 529--597.

\bibitem{S11} Strain, R. M.:
Coordinates in the relativistic Boltzmann theory.
{\it Kinet. Relat. Models} 4 (2011), no. 1, 345--359.

\bibitem{T09} Takou, E.:
Global properties of the solutions of the Einstein-Boltzmann system with
cosmological constant in the Robertson-Walker space-time.
{\it Commum. Math. Sci.} 7 (2009), 399--410.

\bibitem{V} Villani, C.: A review of mathematical topics in collisional kinetic theory.
{\it Handbook of mathematical fluid dynamics, Vol. I,}  71--305,
{\it North-Holland, Amsterdam,} 2002.
\end{thebibliography}
\end{document}